\documentclass{llncs}

\usepackage[utf8]{inputenc}
\DeclareUnicodeCharacter{2200}{$\forall$}
\DeclareUnicodeCharacter{2203}{$\exists$}
\DeclareUnicodeCharacter{2192}{$\rightarrow$}
\DeclareUnicodeCharacter{03BB}{$\mathrm{\lambda}$}
\DeclareUnicodeCharacter{00B9}{$^{1}$}
\DeclareUnicodeCharacter{207B}{$^{-}$}

\usepackage{amsfonts}
\usepackage{mathtools}

\spnewtheorem{axiom}{Axiom}{\bfseries}{\rmfamily}

\usepackage{algorithmicx}
\usepackage[noend]{algpseudocode}
\usepackage{algorithm}

\algrenewcommand\algorithmicrequire{\textbf{Precondition:}}
\algrenewcommand\algorithmicensure{\textbf{Postcondition:}}
\newcommand*\Let[2]{\State #1 $\gets$ #2}
\algnewcommand\andif{\textbf{ and }}

\usepackage{tikz-cd}

\title{Automatic and Transparent Transfer of Theorems along Isomorphisms \\
in the {\sc Coq} Proof Assistant}
\author{Theo Zimmermann\inst{1} \and Hugo Herbelin\inst{2}}
\institute{\'Ecole Normale Sup\'erieure, Paris, France \\ \email{theo.zimmermann@ens.fr}
\and Inria Paris-Rocquencourt, Paris, France \\ \email{hugo.herbelin@inria.fr}}

\begin{document}
\maketitle

\begin{abstract}
In mathematics, it is common practice to have several constructions
for the same objects.
Mathematicians will identify them modulo isomorphism and will
not worry later on which construction they use, as theorems proved
for one construction
will be valid for all.

When working with proof assistants, it is also common to see several
data-types representing the same objects.
This work aims at making the use of several isomorphic constructions
as simple and as transparent as it can be done informally in
mathematics.
This requires inferring automatically the missing proof-steps.

We are designing an algorithm which finds and fills these missing
proof-steps and we are implementing it as a plugin for \textsc{Coq}\footnote{
This plugin introduces a new tactic called \texttt{exact modulo}.
Its most recent version is available on the web at
\url{https://github.com/Zimmi48/transfer}.}.
\end{abstract}

\section{Introduction}

With examples such as the well-known relation between
linear maps and matrices, the various constructions of real numbers
(equivalence classes of Cauchy sequences, Dedekind cuts,
infinite sequences of digits, subset of complex numbers), we see that
there are a great many cases when identifying several constructions
of the same objects can be useful in mathematics. In particular, proofs
are then done on the most convenient one but theorems apply to all.

In formal systems like \textsc{Coq} \cite{Coq:manual}, a canonical
example is the various constructions available for natural numbers.
The most natural construction and the closest to the mathematical view
is unary (\texttt{0}, \texttt{S 0}, \texttt{S (S 0)} and so on) while
the more efficient binary construction is closest to what is available
in most programming languages.

When several constructions coexist, they often share an axiomatic
representation, abstracting away from the internal details.
In \textsc{Coq}, it is possible to do proofs directly on the
axiomatic representation thanks to the module and functor system 
\cite{chrzkaszcz2003implementing}.
While this has the advantage of
factoring proofs, it also makes the proof harder as it does not allow taking
advantage of the specifics of the implementation.

The purpose of this work is to make easy to transport theorems
to all isomorphic constructions even when the proof relies on one
particular such construction. In an informal setting,
the mathematician would declare that ``we can identify the two structures''
once she has proved they were isomorphic and would proceed from there.
Our goal is to justify that claim because it will be that missing
justification that the proof checker will ask for.
Moreover, we need to determine when this justification is missing and
insert it automatically.

Although we focus on isomorphic structures in our description of the
problem and in our examples, we want to emphasize that we thrive to be
as general as possible and require as little as possible to allow
the automatic transfer of a theorem. Sometimes an isomorphism is required
but sometimes a weaker correspondence is sufficient.
Our algorithm will typically allow the following
transfer:

\begin{example}
    Take two sets $A$ and $A'$.
    If we have the following result on the first set:
    \begin{axiom}[A is empty]
        \[\forall x \in A, \bot \enspace .\]
    \end{axiom}
    then a surjective function $f : A \rightarrow A'$ is all we need to
    transfer the result and get:
    \begin{theorem}[A' is empty]
        \[\forall x' \in A', \bot \enspace .\]
    \end{theorem}

    Here is the complete corresponding \textsc{Coq} development
    (using our plugin -- although in that case, it is extremely easy to
    build the proof by hand):
    \begin{verbatim}
    Parameter A A' : Set.
    Axiom emptyA : ∀ x : A, False.
    Parameter f : A → A'.
    Parameter g : A' → A.
    Axiom surjf : ∀ x' : A', f (g x') = x'.
    Declare Surjection f by (g, surjf).
    Theorem emptyA' : ∀ x' : A', False.
      exact modulo emptyA.
    Qed.\end{verbatim}
\end{example}

In the remainder of this text, we will start by presenting our current
algorithm which is able to transfer a limited but already interesting
set of theorems. Then, we will detail our ideas to generalize it.
Finally, we will compare our approach to previous related works.

\section{How to Transfer a Theorem}
\label{sec:transf}

To start, we are limiting ourselves to transferring first-order formulas containing only universal quantifiers,
implication and relations.

\subsection{User-provided declarations}

We only require from the user to provide a set of surjective functions between related data-types,
along with a proof of surjectivity, and transfer lemmas.
That is, we can relate two data-types $A$ and $A'$ by producing a function $f : A \rightarrow A'$
and a proof that $f$ is surjective.
To ease our task, we will require that the proof that $f$ is surjective be given by producing a
right-inverse\footnote{In other words, using
terminology of category theory, we ask
that $g$ be a section of $f$ and $f$ be a retraction of $g$.} $g$ and a proof that
\[
    \forall x' \in A', f( g(x') ) = x' \enspace .
\]
If the user wishes to transfer a relation $R \in A \times A \times \ldots \times A$ to a relation
$R' \in A' \times A' \times \ldots \times A' \enspace$, she must provide a transfer lemma of the
form
\[
    \forall x_1 \ldots x_n \in A, R(x_1, \ldots, x_n) \Rightarrow R'(f(x_1), \ldots, f(x_n))
\]
where $f$ is called the transfer function between $R$ and $R'$.

The declared surjections and transfer lemmas will be stored in tables (maps).
A given surjection can be retrieved by looking for a pair of data-types while a given transfer
lemma can be retrieved by looking for a pair of relations.
There can be only one stored item for each key which prevents defining several distinct isomorphisms
between two structures.

Example \ref{expl:natN} shows how this is enough for transferring interesting theorems from one
data-type to another.

\begin{example}
    \label{expl:natN}
    Suppose we are given two data-types to represent $\mathbb{N}$,
    called nat and N together with two relations
    $\leq_{\mathrm{nat}}$ and $\leq_{\mathrm{N}}$.

    We know nothing of their implementation but we are also given two
    functions $\mathrm{N.to\_nat} : \mathrm{N} \rightarrow \mathrm{nat}$
    and $\mathrm{N.of\_nat} : \mathrm{nat} \rightarrow \mathrm{N}$
    and the four accompanying axioms:

    \begin{axiom}[Surjectivity of N.to\_nat]
        \label{axiom:surj2}
        \[
            \forall x \in \mathrm{nat}, \mathrm{N.to\_nat}( \mathrm{N.of\_nat}(x) ) = x \enspace .
        \]
    \end{axiom}

    \begin{axiom}[Surjectivity of N.of\_nat]
        \label{axiom:surj1}
        \[
            \forall x' \in \mathrm{N}, \mathrm{N.of\_nat}( \mathrm{N.to\_nat}(x') ) = x' \enspace .
        \]
    \end{axiom}

    \begin{axiom}[Transfer from $\leq_{\mathrm{N}}$ to $\leq_{\mathrm{nat}}$ by N.to\_nat]
        \label{axiom:transf2}
        \[
            \forall x', y' \in \mathrm{N},
            x' \leq_{\mathrm{N}} y' \Rightarrow
            \mathrm{N.to\_nat}(x') \leq_{\mathrm{nat}}
            \mathrm{N.to\_nat}(y') \enspace .
        \]
    \end{axiom}

    \begin{axiom}[Transfer from $\leq_{\mathrm{nat}}$ to $\leq_{\mathrm{N}}$ by N.of\_nat]
        \label{axiom:transf1}
        \[
            \forall x, y \in \mathrm{nat},
            x \leq_{\mathrm{nat}} y \Rightarrow
            \mathrm{N.of\_nat}(x) \leq_{\mathrm{N}} \mathrm{N.of\_nat}(y) \enspace .
        \]
    \end{axiom}
    Finally, we are given the following result to transfer:

    \setcounter{theorem}{\value{axiom}}
    \begin{axiom}[Transitivity of $\leq_{\mathrm{nat}}$]
        \label{axiom:transit}
        \[
            \forall x, y, z \in \mathrm{nat},
            x \leq_{\mathrm{nat}} y \Rightarrow
            y \leq_{\mathrm{nat}} z \Rightarrow x \leq_{\mathrm{nat}} z \enspace .
        \]
    \end{axiom}
    All these results enable us indeed to transfer Axiom \ref{axiom:transit}
    into Theorem \ref{thm:transit}.

    \begin{theorem}[Transitivity of $\leq_{\mathrm{N}}$]
        \label{thm:transit}
        \[
            \forall x', y', z' \in \mathrm{N},
            x' \leq_{\mathrm{N}} y' \Rightarrow y' \leq_{\mathrm{N}} z' 
            \Rightarrow x' \leq_{\mathrm{N}} z' \enspace .
        \]
    \end{theorem}
    \begin{proof}
        Let $x', y', z' \in \mathrm{N}$
        and assume that the following two hypotheses hold:
        \begin{equation}
            \label{eqn:hyp1}
            x' \leq_{\mathrm{N}} y' \enspace ,
        \end{equation}
        \begin{equation}
            \label{eqn:hyp2}
            y' \leq_{\mathrm{N}} z' \enspace .
        \end{equation}
        From (\ref{eqn:hyp1}) (respectively (\ref{eqn:hyp2})) and Axiom \ref{axiom:transf2},
        we draw
        \begin{equation}
            \label{eqn:concl1}
            \mathrm{N.to\_nat}(x') \leq_{\mathrm{nat}}
            \mathrm{N.to\_nat}(y') \enspace ,
        \end{equation}
        \begin{equation}
            \label{eqn:concl2}
            \mathrm{N.to\_nat}(y') \leq_{\mathrm{nat}}
            \mathrm{N.to\_nat}(z') \enspace .
        \end{equation}
        We can now apply Axiom \ref{axiom:transit}
        to $\mathrm{N.to\_nat}(x')$,
        $\mathrm{N.to\_nat}(y')$ and $\mathrm{N.to\_nat}(z')$ and conclude
        \begin{equation}
            \label{eqn:concl3}
            \mathrm{N.to\_nat}(x') \leq_{\mathrm{nat}}
            \mathrm{N.to\_nat}(z') \enspace .
        \end{equation}
        We then apply Axiom \ref{axiom:transf1} to get
        \begin{equation}
            \mathrm{N.of\_nat}( \mathrm{N.to\_nat}(x') )
            \leq_{\mathrm{N}} \mathrm{N.of\_nat}( \mathrm{N.to\_nat}(z') ) \enspace .
        \end{equation}
        That is (rewriting with Axiom \ref{axiom:surj1}):
        \begin{equation}
            x' \leq_{\mathrm{N}} z' \enspace .
        \end{equation}
        \qed
    \end{proof}

    You will have noticed that Axiom \ref{axiom:surj2} has not been useful here.
    It would have been if there had been a quantification to transfer inside one of the hypotheses.
    This suggests a similar example where Axiom \ref{axiom:surj2} would not hold,
    thus where there would be no isomorphism between the two related data-types.
    Such an example is provided in the repository containing the plugin: we transfer various
    theorems (such as transitivity of $\leq$) from $\mathbb{Z}$ to $\mathbb{N}$.

\end{example}

\subsection{Preliminaries in type-theory-based logic}

Understanding the proposed algorithm will not require much knowledge about
the internals of Coq:

\begin{itemize}
\item Dependent products are the way in which the Calculus of Inductive Constructions
\cite[Ch. 4]{Coq:manual}, the logical base of \textsc{Coq},
models both universal quantification and implication. The implication is just the
degenerate non-dependent case,
i.e. $A \Rightarrow B$ is just an abbreviation for $\forall x : A, B$ when $x$ does
not appear in $B$.

\item In the Calculus of Inductive Constructions as well as in
any other type-theory-based logic,
proofs can be viewed as programs, and in particular the proof $\rho_{A \Rightarrow B}$
of an implication $A \Rightarrow B$
can be viewed as a function that takes a proof $\rho_A$ of $A$ as argument
and produces a proof $\rho_{A \Rightarrow B} (\rho_A)$ of $B$.
\end{itemize}

\subsection{The algorithm}

Algorithm \ref{alg:transfer-simple} takes as input two formulas (called \emph{theorem} and \emph{goal})
differing only in the data-types that are quantified over and in the relations they contain,
as well as a proof of \emph{theorem}.
It outputs a proof of \emph{goal} provided that the differences between the two formulas all correspond
to previously declared surjections and transfer lemmas.

The algorithm is recursive over the structure of the two formulas (which must be the same).
There are two main cases: when the formulas are atoms (i.e. in our case, relations applied
to arguments) or dependent products.

\begin{algorithm}
  \caption{Transfer a Theorem
    \label{alg:transfer-simple}}
  \begin{algorithmic}[5]
    \Require{In the environment $\mathrm{\Gamma}$, $F$ and $F'$
    are two well-defined formulas \\
    and $\rho_F$ is a proof of $F$.}
    \Ensure{\Call{ExactModulo}{$\mathrm{\Gamma}, F, F', \rho_F$}
    is a proof of $F'$ in environment $\mathrm{\Gamma}$ or
    it is a failure.}
    \Statex
    \Function{ExactModulo}{$\mathrm{\Gamma}, F, F', \rho_F$}
        \If{$F = F'$}
            \State \Return{$\rho_F$}
        \ElsIf{$F = R(t_1, \ldots, t_n) \andif F' = R'(t'_1, \ldots, t'_n)$}
            \Let{$f$}{transfer function between $R$ and $R'$}
            \State \Comment \Return failure if it does not exist
            \Let{$\rho_{\mathrm{transfer}}$}{proof of compatibility of $f$
            with respect to $R$ and $R'$}
            \For{$i \gets 1 \textrm{ to } n$}
                \If{$t'_i \neq f(t_i)$}
                    \State \Return{failure}
                \EndIf
            \EndFor
            \State \Return{$\rho_{\mathrm{transfer}}
            (t_1, \ldots, t_n, \rho_F)$}
        \ElsIf{$F = \forall x : A, B \andif F' = \forall x' : A', B'$}
            \Let{$\mathrm{\Gamma}$}{$\mathrm{\Gamma}, x' : A'$}
            \Let{$t$}{\Call{ExactModulo}{$\mathrm{\Gamma}, A', A, x'$}}
            \If{$t \neq$ failure}
                \Let{$\rho_{\mathrm{rec}}$}{
                \Call{ExactModulo}{$\mathrm{\Gamma}, B, B', \rho_F(t)$}}
                \State \Comment \Return failure if $\rho_{\mathrm{rec}} =$ failure
                \State \Return $\lambda x' : A'\ldotp \rho_{\mathrm{rec}}$
            \Else
                \Let{$f$}{surjection from $A$ to $A'$}
                \Comment \Return failure if it does not exist
                \Let{$g$}{right-inverse of $f$}
                \Let{$\rho_{\mathrm{surjection}}$}{proof that $g$ is a right-inverse of $f$}
                \Let{$B_{\mathrm{subst}}$}{$B$ where $x$ was replaced by $g(x')$}
                \Let{$B'_{\mathrm{subst}}$}{$B'$ where $x'$ was replaced by $f(g(x'))$ in covariant places} \label{line:covariant}
                \Let{$\rho_{\mathrm{rec}}$}{\Call{ExactModulo}{$\mathrm{\Gamma},
                B_{\mathrm{subst}}, B'_{\mathrm{subst}}, \rho_F(g(x'))$}}
                \State \Comment \Return failure if $\rho_{\mathrm{rec}} =$ failure
                \State Now $\lambda x' : A'\ldotp \rho_{\mathrm{rec}}$ is a proof of
                $\forall x' : A', B'_{\mathrm{subst}}$.
                With the help of $\rho_{\mathrm{surjection}}$ we can transform it into $\rho_{F'}$
                a proof of $\forall x' : A', B'$.
                \State \Return $\rho_{F'}$
            \EndIf
        \Else
            \State \Return{failure}
        \EndIf
    \EndFunction
  \end{algorithmic}
\end{algorithm}

You will have noticed, at line \ref{line:covariant} of Algorithm
\ref{alg:transfer-simple}, the strange choice of substituting $x'$ with
$f(g(x'))$ only in covariant places. As $x' = f(g(x'))$, we could
have done the substitution wherever we liked. We do it only in
covariant places so that the formulas in the recursive calls will have
exactly the right form when reaching the atomic case (relations).
One can convince oneself that substituting in covariant
places is enough by observing what it gives on the last example
(transitivity of $\leq_{\mathrm{N}}$)
while remembering that the right-hand side of an
implication is covariant while the left-hand side is contravariant.

We could add support for logical connectives
such as $\wedge$ and $\vee$ or
the existential quantifier $\exists$ but as they play
no specific role in the Calculus of Inductive Constructions
(unlike universal quantification and implication),
we rather want a more general way of treating any such addition.
As for the negation $\neg A$, in \textsc{Coq} it is defined as
$A \Rightarrow \bot$ so it is already supported
provided we unfold its definition first.

\section{Generalizing}

Algorithm \ref{alg:transfer-simple} has quite a lot of limitations at the moment
which we plan to lift.

\paragraph{Functions.}
So far we have considered only relations. Even though any function can be expressed as a relation,
this path would require a lot of preliminary rewriting steps; thus
it would be a lot more convenient to be able to transfer functions directly.
Given that relations are represented as functions to the special sort \texttt{Prop} in \textsc{Coq},
what we need is a generalization where functions to any type, as well as
internal operators, would be supported.

\paragraph{New connectives.}
We want to be able to handle logical connectives such as $\wedge$ and $\vee$ but also various
other combinators and non-propositional functions.
For instance, we should be able to transfer theorems involving equality.

\paragraph{Other equivalence relations.}
Currently, Leibniz (structural) equality plays a special role as it has to appear
in the surjection lemmas. Leibniz equality has the advantage of allowing rewriting in any
subterm. But techniques have already been devised \cite{Sozeau2010}
to allow rewriting with other equivalence
relations and we plan to inspire from them.

\paragraph{No right-inverse.}
For simplicity, we have asked so far
for proofs of surjectivity which involved
producing a right-inverse. This has a major drawback. Indeed,
surjectivity is equivalent to having a right-inverse only
if we admit the Axiom of Choice.
We want our algorithm to be as general as possible,
therefore we will work to remove that requirement.

\subsection{Generalizing Declarations}

\label{subsec:gen_decl}

\subsubsection{Transfer lemmas.}

The \textsc{Coq} Morphisms library\footnote{The \textsc{Coq} Morphisms library is part of the work
of Matthieu Sozeau \cite{Sozeau2010}
to generalize rewriting for equivalence relations that are not Leibniz equality.
Its documentation is available online at \url{https://coq.inria.fr/library/Coq.Classes.Morphisms.html}.}
introduces a new notion of respectful morphisms for a binary homogeneous relation.
We draw from \cite{Cohen2013} the idea of using the generalized heterogeneous version for our transfer
declarations.
Heterogeneous relations bring us the ability to relate objects from one data-type with
objects from another data-type.

We will note
\begin{verbatim}
(R ##> R') f g := ∀ (x : X) (y : Y), R x y → R' (f x) (g y) .\end{verbatim}
This can also be seen as a (commutative) diagram.

\[
\begin{tikzcd}
X \arrow[leftrightarrow]{r}{R} \arrow[swap]{d}{f} & Y \arrow{d}{g} \\
X' \arrow[leftrightarrow]{r}{R'} & Y'
\end{tikzcd}
\]

It is easy to show that this corresponds precisely to a very general notion of homomorphism
that can be found in mathematics textbooks such as \cite[Ch. 5.7]{schmidt2011relational}.
The pair of mappings $(f,g)$ is a homomorphism
between the two ``structures'' $(X \times Y,R)$
and $(X' \times Y',R')$ if the following holds:
\[
    R \circ g \subseteq f \circ R'
\]
where $\circ$ is the relational composition, i.e.
\[
\forall x \in X, y' \in Y', \left[
(R \circ g)(x,y') \Leftrightarrow \exists y \in Y, R(x,y) \wedge g(y) = y'
\right] \enspace ,
\]
\[
\forall x \in X, y' \in Y', \left[
(f \circ R')(x,y') \Leftrightarrow \exists x' \in Y, f(x) = x' \wedge R'(x',y')
\right] \enspace .
\]

It will be possible to declare all sorts of transfer lemmas thanks
to the respectful arrow as can be seen in the following example.

\begin{example}
    Let us consider a heterogeneous binary relation \verb|natN| relating
    elements of \verb|nat| with elements of \verb|N|.
    One possible definition would be:
    \begin{verbatim}Definition natN x x' := N.of_nat x = x'.\end{verbatim}
    Then, we can declare how to transfer various functions and relations:
    \begin{verbatim}Theorem le_transfer : (natN ##> natN ##> impl) le N.le.\end{verbatim}
    where \verb|le| represents $\leq_{\mathrm{nat}}$,
    \verb|N.le| represents $\leq_{\mathrm{N}}$ and \verb|impl|
    is a relation corresponding to the implication (also, note that
    \verb|##>| is right-associative).
    That is, after unfolding the definitions of
    \verb|natN|, \verb|##>| and \verb|impl|:
    \begin{verbatim}Theorem le_transfer :
    ∀ (x : nat) (x' : N), N.of_nat x = x' →
    ∀ (y : nat) (y' : N), N.of_nat y = y' → le x y → N.le x' y'.\end{verbatim}

    Considering two new Boolean functions \verb|iszero_nat| and \verb|iszero_N|,
    we can make explicit how they relate in the following way:
    \begin{verbatim}Theorem iszero_transfer : (natN ##> @eq bool) iszero_nat iszero_N.\end{verbatim}
    where \verb+@eq bool+ is the Boolean equality.

    Finally, considering two operations \verb|Nat.add| and \verb|N.add|:
    \begin{verbatim}Theorem plus_transf : (natN ##> natN ##> natN) Nat.add N.add.\end{verbatim}
\end{example}

\subsubsection{Surjection lemmas.}

That very same idea of respectful morphisms can be used to replace
the surjection declarations we used so far.
Just as we had replaced the implication $\rightarrow$ by a new relation
\texttt{impl}, we will use a new relation \texttt{@all} to represent $\forall~$:
\begin{verbatim}@all A (λ x : A, B) := ∀ x : A, B .\end{verbatim}

Any surjection declaration in the style of Sec. \ref{sec:transf}:
\begin{verbatim}Declare Surjection f by (g, proof).\end{verbatim}
can be equivalently replaced by the following three declarations:
\begin{verbatim}Theorem R_surj : ((R ##> impl) ##> impl) (@all A) (@all A').
Theorem R_tot : ((R⁻¹ ##> impl) ##> impl) (@all A') (@all A).
Theorem R_func : (R ##> R ##> impl) (@eq A) (@eq A').\end{verbatim}
where \verb| R x x' := f x = x' | and \verb| R⁻¹ x' x := R x x' |.

The first declaration corresponds to the surjectivity of relation $R$
(also called right-totality).
The second and third declaration express the fact that $R$ is a mapping.
More precisely, the second declaration corresponds to the surjectivity of the inverse
relation, that is the (left-)totality of $R$.
The third declaration expresses the knowledge that $R$ is functional
(also called univalent in \cite[Ch. 5.1]{schmidt2011relational} or right-unique
elsewhere).

The three declarations provide interesting ``point-free'' formulations
of a relation totality and unicity properties. Let us unfold two of
them to give more intuition on what they mean:

\begin{verbatim}
Theorem R_surj :
    ∀ P P', (∀ (x : A) (x' : A'), R x x' → P x → P' x') →
    (∀ x : A, P x) → ∀ x' : A', P' x'.
Theorem R_func :
    ∀ (x : A) (x' : A'), R x x' →
    ∀ (y : A) (y' : A'), R y y' → x = y → x' = y'.\end{verbatim}
We immediately see that \verb|R_func| indeed expresses that $R$ is
functional (each input has at most one output).
As for \verb|R_surj|, while it is clearly a necessary
condition for surjectivity,
we will have to instantiate the
theorem with \verb|P = λ _ : A, True| and \verb|P' = λ x' : A', ∃ x : A, R x x'|
to see that it is sufficient.

We can already foresee two advantages of this new formulation of
surjectivity lemmas. First, it is more general as it will allow
considering data-types which are related by a non-functional or
non-total relation.
Second, we can already imagine replacing \verb|@eq| by any equivalence
relation and \verb|@all| by any bounded quantification, thus allowing
to relate two partial quotients and not only classic data-types.

\subsection{Transfer to the context}

In \cite{Sozeau2010}, Matthieu Sozeau
gives a set of inference rules to find where a rewrite can
occur and the proof that the rewrite is correct. Building the
proof will sometimes require prior declarations that
some functions are respectful morphisms for
some homogeneous relations.
For our purpose, we need to generalize these rules to
heterogeneous relations.

As before, we take a theorem and a goal
as arguments and we must produce a proof of \texttt{thm → goal}, that is
\texttt{impl thm goal}.
We borrow the notation
\[
    \mathrm{\Gamma} \vdash \tau \leadsto^R_p \tau'
\]
which means that given an environment $\mathrm{\Gamma}$
in which $\tau$ and $\tau'$ are well-defined,
$p$ is a proof of $R(\tau, \tau')$.

Initially, given a theorem $\mathrm{\Gamma} \vdash \tau$ and a goal
$\mathrm{\Gamma} \vdash \tau'$, we want to derive a judgment of the form:
\[
    \mathrm{\Gamma} \vdash \tau \leadsto^{\mathtt{impl}}_p \tau'
\]

\subsubsection{Rules.}

We give in Fig. \ref{fig:rules}
the rules to get to that judgment, adapted from \cite{Sozeau2010}. We have dropped
the \textsc{Unify} rule as it was used for rewriting but does not
apply in our case. To avoid unnecessary complexity,
we have also chosen to drop the \textsc{Sub} rule in a first version.

\begin{figure}
    \[
        \frac{
            p : R(\tau,\tau') \in \mathrm{\Gamma}
        }{
            \mathrm{\Gamma} \vdash
            \tau \leadsto^R_{p} \tau'
        } \mbox{ \textsc{Env} }
        \qquad
        \frac{
            p : R(\tau,\tau') \in \mathrm{Tables}
        }{
            \mathrm{\Gamma} \vdash
            \tau \leadsto^R_{p} \tau'
        } \mbox{ \textsc{Table} }
    \]
    \[
        \frac{
            \mathrm{\Gamma}, x : \tau_1, x' : \tau'_1, H : R(x,x') \vdash
            \tau_2 \leadsto^S_{p} \tau'_2
        }{
            \mathrm{\Gamma} \vdash
            \mathrm{\lambda} x : \tau_1. \tau_2
            \leadsto^{
                R~\#\#>~S
            }_{
                \mathrm{\lambda} x : \tau_1, x' : \tau'_1, H : R(x,x'). p
            }
            \mathrm{\lambda} x' : \tau'_1. \tau'_2
        } \mbox{ \textsc{Lambda} }
    \]
    \[
        \frac{
            \mathrm{\Gamma} \vdash
            f \leadsto^{R~\#\#>~S}_{p_f} f'
            \quad
            \mathrm{\Gamma} \vdash
            e \leadsto^{R}_{p_e} e'
        }{
            \mathrm{\Gamma} \vdash
            f(e) \leadsto^S_{p_f(e, e', p_e)} f'(e')
        } \mbox{ \textsc{App} }
    \]
    \[
        \frac{
            \mathrm{\Gamma} \vdash
            \mathtt{@all}~\tau_1~(\mathrm{\lambda} x : \tau_1. \tau_2)
            \leadsto^R_p
            \mathtt{@all}~\tau'_1~(\mathrm{\lambda} x' : \tau'_1. \tau'_2)
        }{
            \mathrm{\Gamma} \vdash
            \forall x : \tau_1, \tau_2 \leadsto^R_p
            \forall x' : \tau'_1, \tau'_2
        } \mbox{ \textsc{Forall} }
    \]
    \[
        \frac{
            \mathrm{\Gamma} \vdash
            \mathtt{impl}~\tau_1~\tau_2 \leadsto^R_p
            \mathtt{impl}~\tau'_1~\tau'_2
        }{
            \mathrm{\Gamma} \vdash
            \tau_1 \rightarrow \tau_2 \leadsto^R_p
            \tau'_1 \rightarrow \tau'_2
        } \mbox{ \textsc{Arrow} }
    \]
    \caption{\texttt{exact modulo} inference rules.}
    \label{fig:rules}
\end{figure}

From these rules, we plan to derive a deterministic algorithm,
which we will implement and test.

We will now illustrate each of these rules by a few examples,
taken from the transfer of Axiom \ref{axiom:transit}
(transitivity of $\leq_{\mathrm{nat}}$)
to Theorem \ref{thm:transit} (transitivity of $\leq_{\mathrm{N}}$).

\begin{example}
    Initially, we want to find a judgment of the form
    \begin{alignat*}{8}
    \vdash \quad &
        \forall x, y, z \in \mathrm{nat},~&
        x \leq_{\mathrm{nat}} y~& \Rightarrow~&
        y \leq_{\mathrm{nat}} z~& \Rightarrow~& x \leq_{\mathrm{nat}} z & \\
    \leadsto^{\mathtt{impl}} \quad &
        \forall x', y', z' \in \mathrm{N},~&
        x' \leq_{\mathrm{N}} y'~& \Rightarrow~&
        y' \leq_{\mathrm{N}} z'~& \Rightarrow~& x' \leq_{\mathrm{N}} z' &
        \enspace .
    \end{alignat*}
    By rule \textsc{Forall}, this reduces to
    \begin{alignat*}{10}
    \vdash \quad &
        \mathtt{@all}~& \mathrm{nat}~&
        (
        \mathrm{\lambda} x : \mathrm{nat},~&
        \forall y, z \in \mathrm{nat},~&
        x \leq_{\mathrm{nat}} y~& \Rightarrow~&
        y \leq_{\mathrm{nat}} z~& \Rightarrow~& x \leq_{\mathrm{nat}} z
        ) \\
    \leadsto^{\mathtt{impl}} \quad &
        \mathtt{@all}~& \mathrm{N}~&
        (
        \mathrm{\lambda} x' : \mathrm{N},~&
        \forall y', z' \in \mathrm{N},~&
        x' \leq_{\mathrm{N}} y'~& \Rightarrow~&
        y' \leq_{\mathrm{N}} z'~& \Rightarrow~& x' \leq_{\mathrm{N}} z'
        ) \enspace .
    \end{alignat*}
    By rule \textsc{App}, this reduces to
    \begin{alignat}{8}
    \vdash \quad &
        \mathrm{\lambda} x : \mathrm{nat},~&
        \forall y, z \in \mathrm{nat},~&
        x \leq_{\mathrm{nat}} y~& \Rightarrow~&
        y \leq_{\mathrm{nat}} z~& \Rightarrow~& x \leq_{\mathrm{nat}} z
        \nonumber \\
    \leadsto^{R} \quad &
        \mathrm{\lambda} x' : \mathrm{N},~&
        \forall y', z' \in \mathrm{N},~&
        x' \leq_{\mathrm{N}} y'~& \Rightarrow~&
        y' \leq_{\mathrm{N}} z'~& \Rightarrow~& x' \leq_{\mathrm{N}} z'
        \enspace ,
    \label{eqn:lambda}
    \end{alignat}
    \begin{equation}
    \vdash \quad \mathtt{@all}~\mathrm{nat}
    \leadsto^{R~\#\#>~\mathtt{impl}} \quad \mathtt{@all}~\mathrm{N}
    \enspace .
    \label{eqn:all}
    \end{equation}
    Then (\ref{eqn:all}) is solved by applying rule \textsc{Table}.
    We get $R =$ \verb|natN ##> impl| .
    Finally, we can report the value of $R$ in (\ref{eqn:lambda})
    and apply rule \textsc{Lambda} and thus our initial problem reduces
    to
    \begin{alignat*}{8}
    x : \mathrm{nat}, x' : \mathrm{N}, H : \mathtt{natN}~x~x' \vdash\,&
        \forall y, z \in \mathrm{nat},~&
        x \leq_{\mathrm{nat}} y~& \Rightarrow~&
        y \leq_{\mathrm{nat}} z~& \Rightarrow~&
        x \leq_{\mathrm{nat}} z & \\
    \leadsto^{\mathtt{impl}}\,&
        \forall y', z' \in \mathrm{N},~&
        x' \leq_{\mathrm{N}} y'~& \Rightarrow~&
        y' \leq_{\mathrm{N}} z'~& \Rightarrow~& x'
        \leq_{\mathrm{N}} z' &
        \enspace .
    \end{alignat*}

    From now on,
    \begin{alignat*}{3}
    \mathrm{\Gamma} ~=~ &
    x : \mathrm{nat},~ x' : \mathrm{N},~ H : \mathtt{natN}~x~x', & \\ &
    y : \mathrm{nat},~ y' : \mathrm{N},~ H_1 : \mathtt{natN}~y~y', & \\ &
    z : \mathrm{nat},~ z' : \mathrm{N},~ H_2 : \mathtt{natN}~z~z' & \enspace .
    \end{alignat*}
    We now consider the problem of finding a judgment of the form
    \begin{alignat*}{6}
    \mathrm{\Gamma} \vdash \quad &
        x \leq_{\mathrm{nat}} y~& \Rightarrow~&
        y \leq_{\mathrm{nat}} z~& \Rightarrow~& x \leq_{\mathrm{nat}} z & \\
    \leadsto^{\mathtt{impl}} \quad &
        x' \leq_{\mathrm{N}} y'~& \Rightarrow~&
        y' \leq_{\mathrm{N}} z'~& \Rightarrow~& x' \leq_{\mathrm{N}} z' &
        \enspace .
    \end{alignat*}
    By rule \textsc{Impl}, this reduces to
    \begin{alignat*}{6}
    \mathrm{\Gamma} \vdash \quad &
        \mathtt{impl} ~&
        (x \leq_{\mathrm{nat}} y)~&
        (y \leq_{\mathrm{nat}} z~& \Rightarrow~& x \leq_{\mathrm{nat}} z) &
        \\
    \leadsto^{\mathtt{impl}} \quad &
        \mathtt{impl} ~&
        (x' \leq_{\mathrm{N}} y')~&
        (y' \leq_{\mathrm{N}} z'~& \Rightarrow~& x' \leq_{\mathrm{N}} z') &
        \enspace .
    \end{alignat*}
    By rule \textsc{App}, this reduces to
    \begin{equation}
    \mathrm{\Gamma} \vdash
    y \leq_{\mathrm{nat}} z \Rightarrow x \leq_{\mathrm{nat}} z
    \leadsto^{R}
    y' \leq_{\mathrm{N}} z' \Rightarrow x' \leq_{\mathrm{N}} z'
    \enspace ,
    \label{eqn:leq}
    \end{equation}
    \begin{equation}
    \mathrm{\Gamma} \vdash
    \mathtt{impl}~(x \leq_{\mathrm{nat}} y)
    \leadsto^{R \#\#> \mathtt{impl}}
    \mathtt{impl}~(x' \leq_{\mathrm{N}} y')
    \enspace .
    \label{eqn:impl}
    \end{equation}
    By rule \textsc{App}, (\ref{eqn:impl}) reduces again to
    \begin{equation}
    \mathrm{\Gamma} \vdash
    x \leq_{\mathrm{nat}} y \leadsto^{S} x' \leq_{\mathrm{N}} y'
    \enspace ,
    \label{eqn:hyp}
    \end{equation}
    \begin{equation}
    \mathrm{\Gamma} \vdash
    \mathtt{impl} \leadsto^{S \#\#> R \#\#> \mathtt{impl}} \mathtt{impl}
    \enspace .
    \label{eqn:implimpl}
    \end{equation}
    We will make sure that
    the tables are pre-filled so that judgments such as
    (\ref{eqn:implimpl}) can be solved with rule \textsc{Table}.
    In that case, we will get $S =$ \verb|impl⁻¹| and
    $R =$ \verb|impl| .
    Now by rule \textsc{App}, (\ref{eqn:hyp}) reduces to
    \begin{equation}
    \mathrm{\Gamma} \vdash y \leadsto^T y'
    \enspace ,
    \label{eqn:natN}
    \end{equation}
    \begin{equation}
    \mathrm{\Gamma} \vdash \mathtt{le}~x
    \leadsto^{T \#\#> \mathtt{impl^{-1}}} \mathtt{N.le}~x'
    \enspace .
    \label{eqn:leq_partial}
    \end{equation}
    Rule \textsc{Env} allows us to derive (\ref{eqn:natN}) with
    $T = \mathtt{natN}\enspace$.

    As for (\ref{eqn:leq_partial}), it can be solved after a few
    more steps by using the knowledge that
    \verb|(natN ##> natN ##> impl⁻¹) le N.le|, which is equivalent to
    \verb|(natN⁻¹| \verb|##>| \verb|natN⁻¹| \verb|##>| \verb|impl) N.le le|,
    which will be one of the user-provided transfer lemmas
    (it corresponds to Axiom \ref{axiom:transf2}).
    Therefore, there only remains to solve (\ref{eqn:leq}) in ways
    similar to this example.
\end{example}

\section{Related work}

\subsection{Proof reuse}

More than ten years ago, Nicolas Magaud \cite{Magaud2003} proposed
an extension of \textsc{Coq} that seemed to share our objectives.
Notably, he was able to transfer all the theorems that were,
at the time, in the standard Arith library, from \texttt{nat} to
\texttt{N}.

The approach was quite intricate because it was able to transfer
proofs, and not just theorems. Given two isomorphic data-types,
one will be considered as the \emph{origin type} and
the other one as the \emph{target type}.
The first step is to define functions to model the origin constructors
within the target type. Moreover, new recursion operators behaving like
the ones of the origin type are added to the target type.

With such a projection of the origin type into the target type, it is
easy to project operators and relations. Proofs are transferred in the
same way. The last step is to establish extensional equality between
projected operators and the corresponding native operators of the
target type.

While interesting, we do not need to take such a complicated path for
our objective which is only \emph{theorem reuse}. Using Magaud's approach
requires much more work in establishing the relations between the two
data-types. Moreover, our approach is more powerful in a sense: we can
transfer properties between two data-types even if we know nothing of
their content and the transfer lemmas where provided as axioms.

\subsection{Algorithm reuse}

A much more recent work by Cohen et al. \cite{Cohen2013} has been of much
inspiration to us. However, the focus is not the same.
In the context of program verification, the authors
propose a general method for algorithm reuse through
parametricity when refining proof-oriented data-types into efficient
computation-oriented data-types. Parametricity then enables the
automatic transfer of algorithm correctness proofs.
Although they give this general method, they explain why they do not
provide a plugin. Our focus being on transparency and usability by
mathematicians, we decided to create such a plugin.

An other inspiring characteristic of their work lies in that
they typically allow
refined types to contain more objects, including objects which would have
no meaning (no specification).
Although we currently require precisely the opposite so as to be
able to translate theorems stating properties \emph{for all} elements,
including unicity properties, we could quite easily add support for
bounded quantification. Bounded quantification would be useful for
transferring theorems from a subset type to the corresponding elements
of a larger type (for instance from $\mathbb{N}$ to non-negative elements
of $\mathbb{Z}$). Similarly, the new way to declare links between two
data-types presented in Sec. \ref{subsec:gen_decl} makes it easy to use
other equivalence relations than just Leibniz equality.

\subsection{Other works proposing a heterogeneous respectful arrow}

While Cohen et al. \cite{Cohen2013} inspired us to use a
generalized heterogeneous respectful arrow to allow for more precise
transfer declarations and remove the limitations of Algorithm
\ref{alg:transfer-simple},
there are many other (and sometimes older) works proposing
the same definition.
One example of such a work is \cite[Def. 13]{homeier2005design}.
But this is not surprising as we have remarked in Sec. \ref{subsec:gen_decl}
that this arrow just encodes for an already existing mathematical
notion of homomorphism.

Huffman and Kun\v{c}ar \cite{huffman2013lifting} go further as
they also show how the relational unicity and totality properties
can be expressed in terms of the respectful arrow.
They produced a Transfer package for \textsc{Isabelle/HOL} with
comparable objectives to ours, and their \texttt{transfer} tactic
is based on a two-step algorithm
sharing many ideas with Matthieu Sozeau's \cite{Sozeau2010}.
Nothing going as far as their Transfer package has yet been
created for \textsc{Coq}.

\section{Conclusion}

In this paper, we have shown how a simple algorithm can make use of a
few initial declarations to ease the reuse of results from one data-type
to another.

As we improve our algorithm and become able to transfer more theorems,
we will still have a lot to do in order to make our plugin as simple-to-use
as possible.
A first easy step will be to transform our \texttt{exact modulo} tactic
into an \texttt{apply modulo} tactic.
Then, we will need to allow for compositionality
in ways similar to \cite{Cohen2013} and \cite{huffman2013lifting}.
First, by allowing
and handling transfer declarations for parametrized types. Then,
by finding paths from one type to another, even when the relation
between the two was not declared, but can be established by going through
a sequence of transfers.

We view this work as a little but quite interesting step in the enormous task
of making the use of a formal proof system as easy as a
pen-and-paper proof.

\section*{Acknowledgments}

The authors wish to thank the anonymous reviewers for their helpful comments.

\bibliographystyle{plain}
\bibliography{paper}

\end{document}